\pdfoutput=1
\documentclass[a4paper,USenglish,cleveref, autoref, thm-restate,numberwithinsect,nolineno]{socg-lipics-v2021}
\hideLIPIcs

\usepackage{microtype}
\usepackage{mathtools}
\usepackage{amssymb}
\usepackage{amsthm}
\usepackage{bm}
\newcommand{\bbZ}{\mathbb{Z}}
\newcommand{\bbR}{\mathbb{R}}

\newcommand{\bbS}{\mathbb{S}}

\newcommand{\ccH}{\mathcal{H}}
\newcommand{\ccO}{\mathcal{O}}

\newcommand{\AAA}{\mathcal{A}}
\newcommand{\bO}{\mathbf{0}}
\newcommand{\bplus}{\bm{+}}

\newcommand{\ztwo}{\bbZ_2}

\newcommand{\scalar}[2]{#1 \cdot #2}

\DeclarePairedDelimiter{\abs}{\lvert}{\rvert}
\DeclarePairedDelimiter{\norm}{\lVert}{\rVert}

\newcommand{\lineC}{\overline{C}}

\newcommand{\funR}{X}
\newcommand{\funB}{Y}

\newtheorem{fact}[theorem]{Fact}
\theoremstyle{remark}

\newtheorem{problem}{Problem}

\let\epsilon\varepsilon

\DeclarePairedDelimiter{\floor}{\lfloor}{\rfloor}
\DeclarePairedDelimiter{\ceil}{\lceil}{\rceil}

\DeclareMathOperator{\supp}{Supp}
\DeclareMathOperator{\area}{area}

\title{
Eight-Partitioning Points in 3D, and Efficiently Too\thanks{A preliminary version of this work  appeared in \emph{SoCG'24} \cite{socg-version}.}} 

\titlerunning{Eight-partitioning points in 3D}

\author{Boris Aronov}{Department of Computer Science and Engineering, Tandon School of Engineering, New York University, Brooklyn, NY 11201 USA \and \url{https://engineering.nyu.edu/faculty/boris-aronov}}{boris.aronov@nyu.edu}{https://orcid.org/0000-0003-3110-4702}{Work has been supported by NSF grants CCF~15-40656 and CCF~20-08551, and by grant 2014/170 from the US-Israel Binational Science Foundation. Part of this research was conducted while BA was visiting ISTA in the summers of 2022 and 2023. The visit of BA to ISTA in the summer of 2022 was supported by an ISTA Visiting Professorship. Research of BA also partially supported by ERC grant no.~882971, ``GeoScape,'' and by the Erd\H os Center.}
\author{Abdul Basit}{School of Mathematics, Monash University, VIC 3800, Australia \and \url{https://sites.google.com/view/abasit}}{abdul.basit@monash.edu}{https://orcid.org/0000-0003-0754-3203}{Work has been supported by Australian Research Council grant DP220102212.}
\author{Indu Ramesh}{Department of Computer Science and Engineering, Tandon School of Engineering, New York University, Brooklyn, NY 11201 USA.}{ir914@nyu.edu}{https://orcid.org/0009-0008-9967-0819}{Work supported by a Tandon School of Engineering Fellowship and by NSF Grant CCF-20-08551.}
\author{Gianluca Tasinato}{Institute of Science and Technology Austria, Am Campus 1, 3400 Klosterneuburg, Austria \and \url{https://gtasinato.github.io/}}{gianluca.tasinato@ist.ac.at}{https://orcid.org/0009-0008-7231-5753}{}
\author{Uli Wagner}{Institute of Science and Technology Austria, Am Campus 1, 3400 Klosterneuburg, Austria \and \url{https://ist.ac.at/en/research/wagner-group/}}{uli@ist.ac.at}{https://orcid.org/0000-0002-1494-0568}{}

\authorrunning{B. Aronov, A. Basit, I. Ramesh, G. Tasinato, and U. Wagner}

\Copyright{Boris Aronov, Abdul Basit, Indu Ramesh, Gianluca Tasinato, and Uli Wagner} 



\category{}

\relatedversion{} 

\acknowledgements{BA and AB would like to thank William Steiger for insightful initial discussions of the problems addressed in this work.}

\begin{document}
\graphicspath{{./figs/}}
\maketitle

\begin{abstract}
An {\em eight-partition} of a finite set of points (respectively, of a continuous mass distribution) in $\mathbb{R}^3$ consists of three planes that divide the space into $8$ octants, such that each open octant contains at most $1/8$ of the points (respectively, of the mass). In 1966, Hadwiger showed that any mass distribution in $\mathbb{R}^3$ admits an eight-partition; moreover, one can prescribe the normal direction of one of the three planes. The analogous result for finite point sets follows by a standard limit argument.

We prove the following variant of this result: Any mass distribution (or point set) in $\mathbb{R}^3$ admits an eight-partition for which the intersection of two of the planes is a line with a prescribed direction.

Moreover, we present an efficient algorithm for calculating an eight-partition of a set of $n$ points in~$\mathbb{R}^3$ (with prescribed normal direction of one of the planes) in time 
$O^{*}(n^{7/3})$.\end{abstract}

\section{Introduction}

Geometric methods for partitioning space, point sets, or other geometric objects are a central topic in discrete and computational geometry.
Partitioning results are often proved using topological methods and also play an important role in topological combinatorics~\cite{DLGMM,matousek_using_2008,roldan2021survey,Zivaljevic2017Survey}.
A classical example is the famous \emph{Ham-Sandwich Theorem}, which goes back to the work of Steinhaus, Banach, Stone, and Tukey (see \cite[Sec.~1]{roldan2021survey} for more background and references). A ``discrete'' version of this theorem asserts that, given any $d$ finite point sets~$P_1,\dots,P_d$ in~$\bbR^d$, there is an (affine) hyperplane~$H$ that \emph{simultaneously bisects} all $P_i$, i.e., each of the two open half-spaces determined by $H$ contains at most $|P_i|/2$ points, $1\leq i \leq d$. This follows (by a standard limit argument, see \cite[Sec.~3.1]{matousek_using_2008}) from the following ``continuous'' version: Let $\mu_1,\dots,\mu_d$ be \emph{mass distributions} in $\bbR^d$, i.e., finite measures such that every open set is measurable and every hyperplane has measure zero. Then there exists a hyperplane $H$ such that 
$\mu_i(H^+)=\mu_i(H^-)=\frac{1}{2}\mu_i(\bbR^d)$ for $1\leq i\leq d$, where $H^+$ and $H^-$ are the two open half-spaces bounded by $H$.

In this paper, we are interested in another classical equipartitioning problem, first posed by Gr\"unbaum \cite{gru} in 1960: Given a mass distribution (respectively, a finite point set) in $\bbR^d$, can one find $d$ hyperplanes that subdivide $\bbR^d$ into $2^d$ open orthants, each of which contains exactly $1/2^d$ of the mass (respectively, at most $1/2^d$ of the points)? We call such a $d$-tuple of hyperplanes a \emph{$2^d$-partition} of the mass distribution (respectively, of the point set). 

For $d=2$, it is an easy consequence of the planar Ham-Sandwich theorem that any mass distribution (or point set) in $\bbR^2$ admits a four-partition; moreover, the four-partition can be chosen such that one of the lines has a prescribed direction (indeed, start by choosing a first line in the prescribed direction that bisects the given mass distribution; by the Ham-Sandwich Theorem, there exists a second line that simultaneously bisects the two parts of the mass  
on either side of the first line). Alternatively, one can also show that there is always a four-partition such that the two lines are orthogonal. Intuitively, the reason that we can impose such additional conditions is that the four-partitioning problem in the plane is \emph{underconstrained}: A line in the plane can be described by two independent parameters, so a pair of lines have four degrees of freedom, while the condition that the four quadrants have the same mass can be expressed by three equations, leaving one degree of freedom; either one of the additional constraints uses this extra degree of freedom.

In 1966, Hadwiger~\cite{H} gave an affirmative answer to Gr\"unbaum's question for $d=3$ and showed that any mass distribution in $\bbR^3$ admits an eight-partition; moreover, the normal vector of one of the planes can be prescribed arbitrarily. This result was later re-discovered by Yao, Dobkin, Edelsbrunner, and Paterson~\cite{YDEP}.
\begin{theorem}[\cite{H,YDEP}]
\label{thm:hadwiger}
Let $\mu$ be a mass distribution on $\bbR^3$, and let $v \in \bbS^2$.
Then there exists a triple of planes $(H_1, H_2, H_3)$ that form an eight-partition for $\mu$ and such that the normal vector of $H_1$ is $v$.
\end{theorem} 

More recently, Blagojevi\'{c} and Karasev~\cite{blag_karasev_2016} gave a different proof for the existence of eight-partitions and showed the following variant:
\begin{theorem}[\cite{blag_karasev_2016}]
\label{thm:blag_karasev}
Let $\mu$ be a mass distribution on $\bbR^3$.
Then there exists an eight-partition $(H_1, H_2, H_3)$ of $\mu$ such that the plane $H_1$ is perpendicular to both $H_2$ and $H_3$.
\end{theorem}

Our first result is the following alternative version of eight-partitioning, which to the best of our knowledge is new:
\begin{theorem} \label{th:our-8-partition} Given a mass distribution $\mu$ in $\bbR^3$ and a vector $v \in \bbS^2$, there exists an eight-partition $(H_1, H_2, H_3)$ of $\mu$ such that the intersection of the two planes  $H_1$ and $H_2$ is a line in the direction of $v$.
\end{theorem} 

As in the case of the Ham-Sandwich Theorem, each of the three theorems above also implies the existence of the corresponding type of eight-partition for finite point sets, again by a standard limit argument (see Lemma~\ref{lem:limiting_lemma}).

We remark that, in general, $d$ hyperplanes in $\bbR^d$ are described by $d^2$ independent parameters, while the condition that $2^d$ orthants have equal mass can be expressed by $2^d-1$ equations. For $d=3$, this leaves $9-7=2$ degrees of freedom, which allows for any one of the additional conditions imposed in Theorems~\ref{thm:hadwiger}, \ref{thm:blag_karasev}, and \ref{th:our-8-partition}, respectively.
On the other hand, for $d\geq 5$, we have $d^2<2^d-1$, so intuitively Gr\"unbaum's problem is overconstrained. Avis~\cite{av} made this precise and constructed explicit counterexamples using the well-known \emph{moment curve}
$\gamma=\{(t,t^2,\dots,t^d)\colon t\in \bbR\}$ in $\bbR^d$.  
The crucial fact is that any hyperplane intersects the moment curve $\gamma$ in at most $d$ points (\cite[Lemma~1.6.4]{matousek_using_2008}). Thus, for $d\geq 5$, a mass distribution supported on $\gamma$ admits no $2^d$-partition because any $d$ hyperplanes intersect~$\gamma$ in at most $d^2$ points, which subdivide~$\gamma$ into at most $d^2+1$ intervals, hence there are always at least $2^d-d^2-1>0$ orthants that do not intersect~$\gamma$ and hence contain no mass.
The last remaining case $d=4$ of Gr\"unbaum's problem, i.e., the question whether any mass distribution in $\bbR^4$ admits a $16$-partition by four hyperplanes, remains stubbornly open  
(see~\cite{blagojevic_topology_2018}, \cite[Conjecture~7.2]{DLGMM}, \cite[pp.~50--51]{matousek_using_2008}, and \cite[Problem 2.1.4]{roldan2021survey}  for more background and related open problems). 
\medskip

We now turn to the algorithmic question of computing eight-partitions in $\bbR^3$.
\begin{problem} \label{prb:8-partition}
Given a set $P$ of $n$ points in $\bbR^3$, in sufficiently general position, compute three planes $H_1,H_2,H_3$ that form an eight-partition of the points.
\end{problem}

As noted above, the corresponding problem of computing a four-partition of a planar point set can be reduced to finding a Ham-Sandwich cut of two planar point sets that are separated by a line. Megiddo~\cite{Meg} showed that this can be done in linear time.

To characterize the complexity of  Problem~\ref{prb:8-partition}, we introduce the following concept. A \emph{halving line} (resp., \emph{halving plane}) for an $n$-point set in $\bbR^2$ (resp., $\bbR^3$) in general position is a line (resp., plane) that passes through two (resp., three) of the points and divides the remaining ones as equally as possible.
Let $h_2(n)$ (resp., $h_3(n)$) denote the maximum number of halving lines (resp., planes) for an $n$-point set in $\bbR^2$ (resp., $\bbR^3$). The best known upper and lower bounds for $h_2(n)$ are $O(n^{4/3})$, due to Dey~\cite{dey1998}, and $\Omega(n e^{\sqrt{\log n}})$, due to T\'{o}th~\cite{toth2001point}, respectively. For $h_3(n)$, the best-known bounds are are $O(n^{5/2})$, due to Sharir, Smorodinsky, and Tardos~\cite{k-sets-3d}, and $\Omega(n^2 e^{\sqrt{\log n}})$, due to T\'{o}th~\cite{toth2001point}. 

By a result of Lo, Matou{\v{s}}ek, and Steiger~\cite[Proposition~2]{LMS}, a Ham-Sandwich cut of $n$ points in $\bbR^3$ can be computed in time $O^*(h_2(n)) = O^*(n^{4/3})$ (see also \cite{halperin2017arrangements}), where the $O^*(\cdot)$-notation suppresses polylogarithmic factors. 
However, computing eight-partitions in $\bbR^3$ appears to be significantly more difficult. Unlike the planar four-partition problem, there is no known way of reducing it to the computation of a Ham-Sandwich cut. In particular, given two planes $H_1$ and $H_2$ that four-partition a finite point set~$P$ in~$\bbR^3$ (in the sense that every one of the four open quadrants determined by $H_1$ and $H_2$ contains at most $|P|/4$ points), there generally need not exist a third plane $H_3$ such that $H_1,H_2,H_3$ form an eight-partition.

We note that, for fixed dimension $d \geq 3$, the best known algorithm for computing Ham-Sandwich cuts in~$\bbR^d$ runs in time $O(n^{d - 1 - \alpha_d})$ where $\alpha_d > 0$ is a constant depending only on~$d$~\cite{LMS}. When the dimension is part of the input, a decision variant of the problem becomes computationally hard, see, e.g., \cite{knauer2011computational}.

A brute-force algorithm that checks every triple of halving planes solves Problem \ref{prb:8-partition} in time comparable to $O(h_3(n)^3) = O(n^{15/2})$. Yao et al.~\cite{YDEP} and Edelsbrunner \cite{edelsbrunner1986edge} gave a $O(n^6)$-time algorithm that computes an eight-partition (with a prescribed normal direction for one of the planes, as in Theorem~\ref{thm:hadwiger}) by an exhaustive search, using the fact that only two planes need to be identified.  Fixing one plane and performing a brute-force search for the remaining two would yield an algorithm with a running time comparable to $O(h_3(n)^2) = O(n^5)$.

Here, we present, to our knowledge, the fastest known algorithm for Problem~\ref{prb:8-partition}.  
Roughly speaking, our algorithm runs in time near-linear in $h_3(n)$ rather than quadratic in it.  Slightly more precisely, our algorithm runs in time near-linear in $nh_2(n)$, which is not known to be $o(h_3(n))$, but for which the best known upper bound is strictly stronger; see Theorem~\ref{thm:alg-detailed} and Fact~\ref{fact:intersectioncomplexity2}:
\begin{theorem}[Algorithm]       
  \label{th:our-algorithm}
  An eight-partition of $n$ points in general position in~$\bbR^3$, with a prescribed normal vector for one of the planes, can be computed in time $O^*(nh_2(n))$, hence $O^*(n^{7/3})$; here, the $O^*(\cdot)$-notation suppresses polylogarithmic factors.
\end{theorem}
Our algorithm can be seen as a constructive version of Hadwiger's proof \cite{H}. We start by bisecting the point set by a plane with a fixed normal direction, which partitions the initial point set into two subsets of ``red'' and ``blue'' points, respectively, of equal size. After that, our algorithm finds two more planes that simultaneously four-partition both the red and the blue points.

It remains an open question whether Theorem~\ref{thm:blag_karasev} or our own Theorem~\ref{th:our-8-partition} can also be used to obtain an efficient algorithm for Problem \ref{prb:8-partition}. It would also be interesting to decide whether there is an algorithm for Problem~\ref{prb:8-partition} with running time $o(nh_2(n))$.
 
\section{The topological result}
\label{sec:topology}

\subsection{Notation and preliminaries}
In what follows, it will often be convenient to assume that the mass distributions we work with have \emph{connected support}, where the support of a mass distribution~$\mu$ is $\supp(\mu) \coloneqq \{x\in \bbR^3 : \mu(B_r(x))> 0 \text{ for every }r>0 \}$ and $B_r(x)$ denotes the ball of radius $r>0$ centered at $x$.

By a standard limit argument (see Lemma~\ref{lem:con_support}), the existence of eight-partitions for mass distributions with connected support implies the existence of eight-partitions for the general case. Hereafter, unless stated otherwise, we assume, without loss of generality, that every mass distribution has connected support.

We denote the \emph{scalar product} of two vectors $x, y \in \bbR^3$ by $\scalar{x}{y}\coloneqq\sum_{i=1}^3 x_iy_i.$
A vector $v \in \bbR^3 \setminus \{ \bO \}$ and a scalar $a \in \bbR$ determine an (affine) plane
\[
  H = H_v(a) := \{ x \in \bbR^3: \scalar{x}{v} = a\},
\]
together with an orientation of $H$ (given by the direction of the normal vector $v$). We denote by $-H := H_{-v}(-a)$ the affine plane with the same equation as $H$ but with opposite orientation. The oriented 
plane $H$ determines two open half-spaces, denoted by
\[
  H^+ := \{ x \in \bbR^3: \scalar{x}{v} > a \} \quad\text{and}\quad H^- := \{ x \in \bbR^3: \scalar{x}{v} < a \}.
\]
More generally, let $\ccH = (H_1, \dots, H_k)$ be an ordered $k$-tuple of (oriented) planes in $\bbR^3$, $k\leq 3$. In what follows, it will be convenient to identify the set $\{+,-\}$ with the group $\bbZ_2$ (where the group operation is multiplication of signs). Elements of $\{+,-\}^k=\bbZ_2^k$ are strings of signs of length $k$, and we will denote by $\bplus = +\dots+$ the identity element of $\bbZ_2^k$.

For $\alpha=(\alpha_1,\dots,\alpha_k) \in \bbZ_2^k= \{+,-\}^k$, we define the \emph{open orthant} determined by $\ccH$ and $\alpha$ as
  $\ccO_\alpha^\ccH := H_1^{\alpha_1} \cap \dots \cap H_k^{\alpha_k}$.
Given a mass distribution $\mu$ in $\bbR^3$, we say that an ordered $k$-tuple $\ccH = (H_1, \dots, H_k)$ of planes ($k \leq 3$) forms a \emph{$2^k$-partition} of $\mu$ if every orthant contains $1/2^k$ of the mass, i.e., $\mu(\ccO_\alpha^{\ccH}) =  \mu(\bbR^3)/2^k$ for every $\alpha \in \{+,-\}^k$. For $k=1,2,3$, this corresponds to the notions of bisecting, four-partitioning, and eight-partitioning $\mu$ as mentioned in the introduction. Analogously, we say that $\ccH$ forms a \emph{$2^k$-partition} of a \emph{finite} point set $P$ in $\bbR^3$ if $\lvert P\cap \ccO^\ccH_\alpha\rvert \leq \frac{\lvert P\rvert }{2^k}$ for all $\alpha$.

We will parameterize oriented planes in $\bbR^3$ by $\bbS^3$, where the north pole $e_4$ and the south pole $-e_4$ map to the plane at infinity with opposite orientations. 
For this we embed $\bbR^3$ into $\bbR^{4}$ via the map $ (x_1,x_2, x_3) \mapsto (x_1,x_2,x_3,1).$ An oriented plane in $\bbR^3$ is mapped to
an oriented affine $2$-dimensional subspace of $\bbR^{4}$ and is extended (uniquely)
to an oriented linear hyperplane. The unit normal vector on the positive side of the linear hyperplane defines a point on the sphere $\bbS^3$. Hence, there is a one-to-one correspondence between points $v$ in $\bbS^3 \setminus \{e_{4}, -e_{4}\}$ and oriented affine planes $H_v$ in $\bbR^3$. The positive side of the plane at infinity is $\bbR^3$ for $v=e_4$ and $\emptyset$ for $v =-e_4$. Hence $H_{-v}^+ = H_v^-$ for every $v$.
Note that planes at infinity cannot arise as solutions to the measure partitioning problem, since they produce empty orthants. Therefore we do not need to worry about the fact that the sphere includes these.

We parameterize triples of planes (called \emph{plane configurations}) in $\bbR^3$ by $(\bbS^3)^3$, and denote by $\ccH_v$ the triple corresponding to $v \in (\bbS^3)^3$.
Given a mass distribution $\mu$ on $\bbR^3$, for each $v \in (\bbS^3)^3$ and $\alpha \in \ztwo^3\setminus\{\bplus\}$, we set
\[ F_\alpha(v, \mu) = \sum_{\beta \in\bbZ_2^3} (-1)^{p(\alpha,\beta)}\mu(\ccO_\beta^{\ccH_v}).
\]
where $p(\alpha, \beta)$ is the number of coordinates where both $\alpha$ and $\beta$ are $-$. 
The functions $F_\alpha$ were also utilized in the proof of Theorem~\ref{thm:hadwiger} in \cite{YDEP}.

As an example, with $\ccH := \ccH_v = (H_1, H_2, H_3)$ and $\alpha = --+ \in \bbZ_2^3 \setminus \{\bplus\}$, we obtain
\begin{align*}
  F_{--+}(\ccH, \mu) & = \sum_{\beta \in\bbZ_2^3} (-1)^{p(\alpha,\beta)}\mu(\ccO_\beta^{\ccH})
  = \sum_{\beta \in\bbZ_2^3:\,p(\alpha, \beta) = 0} \mu(\ccO_\beta^{\ccH}) - \sum_{\beta \in\bbZ_2^3:\,p(\alpha, \beta) = 1} \mu(\ccO_\beta^{\ccH})\\
                     & = \left(\strut\mu(\ccO^{\ccH}_{+++}) + \mu(\ccO^{\ccH}_{++-}) +  \mu(\ccO^{\ccH}_{--+}) + \mu(\ccO^{\ccH}_{---})\right)\\
                     & \quad - \left(\strut\mu(\ccO^{\ccH}_{-++}) + \mu(\ccO^{\ccH}_{-+-}) +  \mu(\ccO^{\ccH}_{+-+}) + \mu(\ccO^{\ccH}_{+--})\right)\\
                     & = \mu(H_1^+ \cap H_2^+) + \mu(H_1^- \cap H_2^-) - \mu(H_1^- \cap H_2^+) - \mu(H_1^+ \cap H_2^-).
\end{align*}
When $\mu$ is clear from context,  we write $F_\alpha(\ccH)$ instead of $F_{\alpha}(\ccH, \mu)$. The definitions of alternating sums for a pair of planes or a single plane are analogous. 

The alternating sums have the following properties which will play an important role in the proof below.
\begin{observation} \label{obs:alt_sums}
  Let $\mu$ be a mass distribution and fix $k=2,3$.
  \begin{enumerate}[(i)]
    \item\label{count:trivial} Let $\alpha \in \ztwo^{k-1}\setminus\{\bplus\}$ and let $\ccH = (H_1, \dots, H_k)$ be a $k$-tuple of planes. Then $F_{+\alpha}(\ccH) = F_{\alpha}((H_2,\dots, H_k))$ (the equivalent statement holds for any other entry of a $k$-tuple $(\alpha_1, \cdots, \alpha_k)$ instead of just for $\alpha_1$).
    \item\label{count:iff} A $k$-tuple $\ccH$ of planes $2^k$-partitions if and only if $F_\alpha(\ccH) = 0$ for every $\alpha \in \ztwo^k\setminus\{\bplus\}$.
  \end{enumerate}
\end{observation}

    \begin{proof}
      (\ref{count:trivial}) 
      Since every hyperplane has null measure it follows that, for any $\beta\in \ztwo^{k-1}$
      \[
        \mu(\ccO^{(H_2, \dots, H_k)}_{\beta}) = = \mu(\ccO^{\ccH}_{+\beta})+\mu(\ccO^{\ccH}_{-\beta}).
      \]

      If $\tilde\alpha=+\alpha$ then, from the definition of $F_{\tilde\alpha}$, for any $\beta\in \ztwo^{k-1}$ the two orthants $\ccO^\ccH_{+\beta}$ and $\ccO^\ccH_{-\beta}$ are counted with the same sign in the sum, therefore
      \[
        F_{\tilde\alpha}(\ccH) = \sum_{\substack{\beta\in \ztwo^{k-1}\\ p(\alpha, \beta) = 0}} \left(\mu(\ccO^\ccH_{+\beta}) + \mu(\ccO^\ccH_{-\beta})\right) - \sum_{\substack{\beta\in \ztwo^{k-1}\\ p(\alpha, \beta) = 1}} \left(\mu(\ccO^\ccH_{+\beta}) + \mu(\ccO^\ccH_{-\beta})\right) = F_{\alpha}(\left(H_2, \dots, H_k\right)).      \]

\smallskip

      \noindent
      (\ref{count:iff}) It is clear that, if $\ccH$ is a $2^k$-partition, then all the alternating sums are $0$. We will prove the other implication. 
      
      Suppose first that $k=1$ and that $\ccH = H$. The only alternating sum is $ F_-(H) = \mu(H^+) - \mu(H^-)$ and $F_-(H) = 0$ implies that $H$ bisects $\mu$. 

      Suppose now that $k=2$ and that $\ccH = (H_1, H_2)$. By (\ref{count:trivial}) and the statement for a single plane, $F_{+-}(\ccH) = 0$ and $F_{-+}(\ccH) = 0$ imply that both $H_1$ and $H_2$ bisect. Therefore, if $\lambda \coloneqq \mu(\ccO^\ccH_{++})$, we have that
      \[
        0 = F_{--}(\ccH) = \mu(\ccO^\ccH_{++}) + \mu(\ccO^\ccH_{--}) - \mu(\ccO^\ccH_{-+}) - \mu(\ccO^\ccH_{+-}) = \lambda +\lambda -(\frac12 - \lambda) - (\frac12 -\lambda) = 4\lambda -1;
      \]
      hence $\lambda = \frac14$ as desired.

      Finally, suppose that $k=3$ and that $\ccH = (H_1, H_2, H_3)$. By (\ref{count:trivial}) and the statement for single planes and for pairs of planes, we have that all planes $H_i$ bisect and all pairs $(H_i, H_j)$ four-partition. Therefore, if $\lambda \coloneqq \mu(\ccO^{\ccH}_{+++})$, we have that
      \begin{align*}
        0 = F_{---}(\ccH) &= \mu(\ccO^{\ccH}_{+++}) + \mu(\ccO^{\ccH}_{+--}) +\mu(\ccO^{\ccH}_{-+-}) + \mu(\ccO^{\ccH}_{--+})\\
        & \quad - \mu(\ccO^{\ccH}_{-++}) -\mu(\ccO^{\ccH}_{+-+})-\mu(\ccO^{\ccH}_{++-})-\mu(\ccO^{\ccH}_{---})\\
        &= \lambda + \lambda +\lambda +\lambda -(\frac14 -\lambda ) - (\frac14 -\lambda ) -(\frac14 -\lambda ) -(\frac14 -\lambda ) = 8\lambda -1;
      \end{align*}
      hence $\lambda =\frac18$ as claimed.
    \end{proof}

\subsection{The main topological result}

Our goal is to prove the following result, which is a more precise statement of Theorem~\ref{th:our-8-partition}: 
\begin{theorem}
\label{thm:top-detailed}
  Given a mass distribution $\mu$ and a direction $u \in \bbS^2$, there exists a triple $\ccH=(H_1, H_2, H_3)$ of oriented planes that eight-partition $\mu$ so that the oriented direction of the intersection $H_1\cap H_2$ is $u$.
\end{theorem}

By Lemma~\ref{lem:con_support}, it is sufficient to prove Theorem~\ref{thm:top-detailed} for mass distributions with connected support. We require the following technical lemma about partitioning a mass distribution on $\bbR^2$, due to Blagojevi\'c and Karasev~\cite{blag_karasev_2016}. For completeness, the proof is given in Appendix~\ref{app:planepartitions}.
\begin{restatable}[Four-partitioning a mass distribution in $\bbR^2$~\cite{blag_karasev_2016}]{lemma}{blagkarasev}
\label{lem:2dpartition}
    Let $\mu^\#$ be a mass distribution (with connected support) on $\bbR^2$ and $v\in \bbS^1$. 
    Then there exists a pair $(\ell_1, \ell_2)$ of lines in $\bbR^2$ 
    that four-partitions $\mu^\#$ and such that $v$ bisects the angle between $\ell_1$ and~$\ell_2$.

    Moreover, if we orient $\ell_1$ and $\ell_2$ so that $\ell_1$ is in the first direction 
    clockwise from $v$, and $\ell_2$ is in the first direction counterclockwise, the oriented pair is unique and the lines depend continuously on $v$.
\end{restatable}
\begin{proof}[Proof of Theorem~\ref{thm:top-detailed}]
  Without loss of generality, let $u = (0,0,1)$. Our proof proceeds in two steps. In the first step, we construct a map $\Phi \colon \bbS^1\times \bbS^3 \rightarrow \bbR^4$ whose zeros codify eight-partitions of $\mu$; then we prove that $\Phi$ is equivariant with respect to a suitable choice of actions of $G := \bbZ_4\times \bbZ_2$ on the two spaces.
  In the second step we show that any continuous $G$-equivariant map $\Psi\colon \bbS^1\times \bbS^3 \rightarrow \bbR^4$ has to have a zero.
\medskip

\noindent
{\bf Step 1: }
  The key step in constructing the map $\Phi$ is to show that we can parameterize pairs of planes that have intersection direction $u$ and four-partition $\mu$, by a vector in $\bbS^1$.
  
  We project $\mu$ onto $u^\bot$, the plane orthogonal to $u$, to obtain a mass distribution $\mu^\#$ on $\bbR^2$. Specifically, identifying the plane with $\bbR^2$, let
  $A \subseteq \bbR^2$ and set $A\times\bbR $ to be the cylinder over $A$ in the $u$-direction. Then $\mu^\#(A) = \mu(A \times \bbR)$.
  
  Let $v\in \bbS^1\subseteq \bbR^2$. By Lemma~\ref{lem:2dpartition}, there are two oriented lines $\ell_1(v) $ and $\ell_2(v)$ (that we can interpret as points $(\ell^i_1, \ell^i_2, a^i)\in \bbS^2$) in the plane $u^\bot$ such that $v$ bisects the angle between the two and $\ell_1, \ell_2$ four-partition the projected measure~$\mu^\#$.
   Define $H_i(v) \coloneqq(\ell_1^i(v),\ell^i_2(v), 0, a^i(v))\in \bbS^3$ to be $\ell_i(v)\times \bbR$ the (oriented) span of $\ell_i(v)$ and $u$; the two planes now four-partition $\mu$ and have the desired intersection direction.

  Now let $g_1$ be a generator of $\bbZ_4 \times \{\bplus\}\subseteq G$ and define its action on $\bbS^1$ by a counterclockwise rotation by $\frac{\pi}{2}$. We use $g_1 \cdot v$ to denote the action of $g_1$ on $v$. Then, by the uniqueness in Lemma~\ref{lem:2dpartition}, we have that (see Figure~\ref{fig:2d_example}):
  \begin{equation}\label{eq:order4action}
    \vec{\ell}_1(g_1\cdot v) = \vec{\ell}_2(v) \quad\text{and}\quad
    \vec{\ell}_2(g_1\cdot v) = -\vec{\ell}_1(v).
  \end{equation}

  \begin{figure}
    \centering
    \includegraphics[width=.75\textwidth]{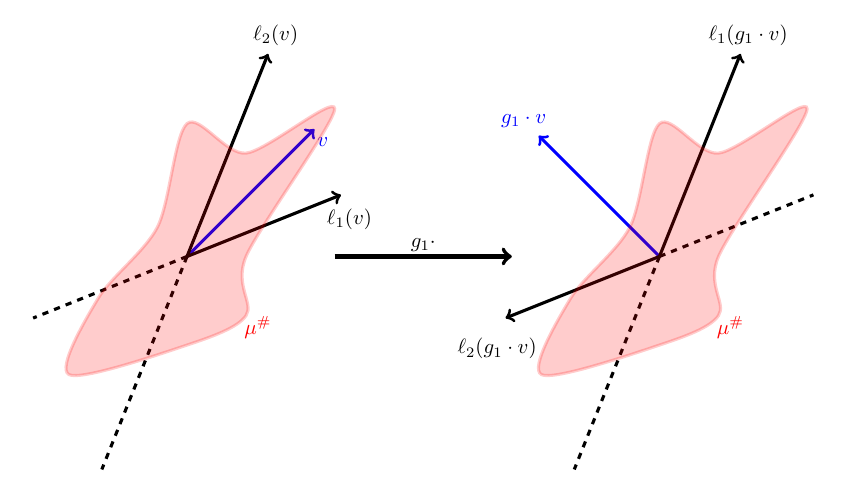}
    \caption{Example of the action of $g_1$.
    }
    \label{fig:2d_example}
  \end{figure}

  Using this construction, we can define a function $\bbS^1\rightarrow \bbS^3\times \bbS^3$ by $v \mapsto (H_1(v), H_2(v))$. It follows from eq.~\eqref{eq:order4action} that $g_1\cdot v$ is mapped to $(H_2(v), -H_1(v))$. Therefore, if we fix the corresponding action\footnote{Formally, for any $(x, y)\in \bbS^3\times \bbS^3$ the generator $g_1$ of $\bbZ_4\times \{+\}\subseteq G$ acts by $g_1\cdot(x, y) = (y, -x)$.} of $\bbZ_4$ on $\bbS^3\times \bbS^3$, the map is $\bbZ_4$-equivariant. 

  The group $\{e\}\times \ztwo$ acts by antipodality on $\bbS^3$; therefore, if $G$ acts on $\left(\bbS^3\times \bbS^3\right)\times \bbS^3$ component-wise, the map $\Phi \colon \bbS^1\times \bbS^3\to \left(\bbS^3\times \bbS^3\right)\times \bbS^3$  defined as $\Phi(v, w) \coloneqq (H_1(v), H_2(v), H_w)$ is $G$-equivariant. 

  By construction, the first two planes are always a four-partition of the mass distribution, therefore by Observation~\ref{obs:alt_sums}, a configuration $\Phi(v, w)$ is an eight-partition if and only if the four alternating sums with $\alpha_3=-$ (i.e., $\alpha = ++-$, $-+-$, $+--$ and $---$) are $0$.

  To compute the action of $G$ on the alternating sums, it is enough to specify what happens on $g_1$ (a generator of $\bbZ_4\times \{e\}$) and $g_2$ (the generator of $\{e\}\times \ztwo$). Recalling that $g_1 \cdot \Phi(v, w) = (H_2(v), -H_1(v), H_w)$ and applying Observation~\ref{obs:alt_sums}(\ref{count:trivial}), we obtain
  \begin{align*}
  & ~~~~ F_{+--}(g_1\cdot \Phi(v,w))  = F_{+--}((H_2(v), -H_1(v), H_w)) = F_{--}((-H_1(v), H_w)) \\
  & = \mu(-H_1(v)^+ \cap H_w^+) + \mu(-H_1(v)^- \cap H_w^-) - \mu(-H_1(v)^- \cap H_w^+) - \mu(-H_1(v)^+ \cap H_w^-)\\
  & = \mu(H_1(v)^- \cap H_w^+) + \mu(H_1(v)^+ \cap H_w^-) - \mu(H_1(v)^+ \cap H_w^+) - \mu(H_1(v)^- \cap H_w^-)\\
  & = - F_{-+-}(\Phi(v,w))
  \end{align*}
  Similar computations imply that, if we act with $g_1$, we obtain
  \begin{align*}
    F_{++-}(g_1\cdot \Phi(v,w)) &=  F_{++-}(\Phi(v,w)),\\
    F_{-+-}(g_1\cdot \Phi(v,w)) &=  F_{+--}(\Phi(v,w)), \text{ and}\\
    F_{---}(g_1\cdot \Phi(v,w)) &=  - F_{---}(\Phi(v,w)),
 \end{align*}
while acting with $g_2$ produces
 \begin{align*}
    F_{++-}(g_2\cdot \Phi(v,w)) &=  - F_{++-}(\Phi(v,w)),\\
    F_{+--}(g_2\cdot \Phi(v,w)) &=  - F_{-+-}(\Phi(v,w)),\\
    F_{-+-}(g_2\cdot \Phi(v,w)) &=  - F_{+--}(\Phi(v,w)), \text{ and}\\
    F_{---}(g_2\cdot \Phi(v,w)) &=  - F_{---}(\Phi(v,w)),
  \end{align*}
  for every $\left(v, w\right)\in \bbS^1\times \bbS^3$.

  Finally, we can choose a linear $G$-action on $\bbR^4$ that is consistent with the previous equations. In particular, if we define
  \[
    g_1\cdot (x, y, z, u) = (x, -z, y, -u) \quad\text{and}\quad
    g_2 \cdot (x, y, z, u) = (-x, -y, -z, -u),
  \]
  then the map $\Psi \colon \bbS^1\times \bbS^3 \rightarrow \bbR^4$, given by
  \[
    (v, w) \mapsto \left(F_{++-}(v,w), F_{+--}(v,w), F_{-+-}(v,w), F_{---}(v,w)\right)
  \]
  is $G$-equivariant, i.e., $\Psi(g_1\cdot v, w) = g_1\cdot \Psi(v, w)$ and $\Psi(v, -w) = -\Psi(v, w)$. By Observation~\ref{obs:alt_sums}, the zeros of $\Psi$ are exactly the configurations of planes that eight-partition the measure and have the desired intersection property.

\medskip

\noindent
  {\bf Step 2: }
  Suppose now, for a contradiction, that $\Psi$ does not have a zero. This means that it is possible to define a $G$-equivariant map $\overline{\Psi} \colon \bbS^1\times \bbS^3 \rightarrow \bbS^3$ by $ \overline{\Psi}(v, w) \coloneqq \frac{\Psi(v, w)}{\norm{\Psi(v, w)}}$.

  Denote by $\Psi_a$, for $a\in \bbS^1$, the map $\Psi_a\colon \bbS^3 \rightarrow \bbS^3$, $\Psi_a(w) = \overline{\Psi}(a, w)$; this function has two key properties:
  \begin{enumerate}[(i)]
  \item for any $a \in \bbS^1$, $\Psi_a$ is antipodal;
  \item for any $a, b \in \bbS^1$, let $\gamma: [0,1]\to \bbS^1$ be a parametrization of the arc between $\gamma(0) = a$ and $\gamma(1)=b$. Then $\Psi_a$ and $\Psi_b$ are homotopic via $H: \bbS^3\times [0,1]\to \bbS^3$ with $H(t, x) = \Psi_{\gamma(t)}(x)$.
  \end{enumerate}

   For any $n\geq 1$, the group of orthogonal matrices $O(n)$ contains exactly two connected components, distinguished by the sign of the determinant. Since the map $g_1:S^3\to S^3$ is induced by a matrix with determinant $-1$, it is homotopic to any other orthogonal linear map with determinant $-1$. In particular, it is homotopic to the reflection $r$ along the last coordinate and, thus, $\deg(g_1) = \deg(r) = -1$.
   Combining everything together we have:
  \[
    \deg(\Psi_a) = \deg(\Psi_{g_1\cdot a})= \deg(g_1\cdot \Psi_a) = \deg(g_1)\deg( \Psi_a) = - \deg(\Psi_a).
  \]
  Hence $\deg(\Psi_a) = 0$, contradicting the Borsuk-Ulam theorem (see \cite[Theorem 2.1.1]{matousek_using_2008}).
\end{proof}

Theorem~\ref{thm:top-detailed}, along with Lemma~\ref{lem:limiting_lemma}, immediately implies the following.
\begin{theorem}
    Let $P \subseteq \bbR^3$ be a finite set of points and $p\in \bbS^2$ a fixed direction. Then there exists a triple $\ccH = (H_1, H_2, H_3)$ of oriented planes that eight-partitions $P$, so that the oriented direction of the intersection $H_1\cap H_2$ is $p$.
\end{theorem}

\section{Levels in arrangements of planes}
\label{sec:levels}

Let $P \subseteq \bbR^3$ be a set of $n$ points in general position. Specifically, we assume that the points in $P$ satisfy the following: no four in a plane, no three in a vertical plane, and no two in a horizontal plane.
Recall that a \emph{halving plane} for a point set in $\bbR^3$ in general position is a plane that passes through three of the points and divides the remaining points as equally as possible; $h_3(n)$ is the maximum number of halving planes for an $n$-point set in $\bbR^3$.  

The duality transform maps points in $\bbR^3$ to planes in $\bbR^3$ and vice versa. Specifically, the point $p = (p_1, p_2, p_3) \in \bbR^3$ is mapped to the non-vertical plane $p^* \colon z = p_1x + p_2y - p_3$ in~$\bbR^3$, and vice versa. See \cite[Chapter 25.2]{har2011geometric} for standard properties of the duality transform.

Let $H$ be a set of planes in $\bbR^3$ in general position. Specifically, we assume that the planes are non-vertical, every triple of planes in $H$ meets in a unique point, and no point in $\bbR^3$ is incident with more than three of the planes.
The planes in $H$ partition $\bbR^3$ into a complex of convex cells, called the \emph{arrangement} of $H$ and denoted by $\AAA(H)$. The \emph{$k$-level} in $\AAA(H)$ is defined as the closure of the set of all points which lie on a unique plane of the arrangement and have exactly $k-1$ planes below it. Note that the $k$-level is a piecewise linear surface in~$\bbR^3$ whose faces are contained in planes of $H$. The \emph{complexity} of a level is the total number of vertices, edges and faces contained in the level.
When $k = \floor{(|H|+1)/2}$, the $k$-level is called the \emph{median level} of the arrangement.

Duality, $k$-levels, and complexity of levels are defined analogously in $\bbR^2$. 
Let $g_2(n)$ (resp., $g_3(n)$) be the maximum complexity of any $k$-level in an arrangement of $n$ planes in~$\bbR^2$ (resp., $\bbR^3$).
It is well-known that $h_2(n) = \Theta(g_2(n))$ and that $h_3(n)=\Theta(g_3(n))$
(see~\cite[Theorem~3]{andrzejak1998results}, \cite[Theorem~3.3]{E}).

The main object of interest in bounding the complexity of our algorithm is the intersection of median levels of two arrangements of disjoint sets of planes. We first show that the complexity of this is proportional to $h_3(n)$, in the worst case.

\begin{fact}
  \label{fact:intersectioncomplexity}
  Let $\ell(n)$ be the maximum complexity of the intersection of a level in the arrangement $\AAA(R)$ and a level in the arrangement $\AAA(B)$ of disjoint sets~$R$ and~$B$ of planes in~$\bbR^3$, with $H\coloneqq A\cup B$ in general position and $n\coloneqq\abs{H}$.  
  Then $\ell(n)=\Theta(h_3(n))$.

\end{fact}
\begin{proof}
  Let $L$ be the intersection of level~$k$ of~$\AAA(R)$ and level~$k'$ of~$\AAA(B)$. As the planes of $H$ are in general position, $L$ is (a disjoint union of) a collection of (open) edges and vertices in $\AAA(R\cup B)$. In fact, $L$ is a collection of cycles and bi-infinite curves so its complexity is asymptotically determined by the number of its edges.

  A point on an edge of $L$ has the property that it lies on one plane of $R$ and one plane of $B$, and has exactly $k+k'-2$ planes of $R \cup B$ below it. 
  Hence $L$ is contained in the $(k+k'-1)$-level of $\AAA(R \cup B)$.  It follows that the complexity of $L$ is bounded above by the complexity of a level in an arrangement of $n$ planes, implying $\ell(n) \leq g_3(n) = O(h_3(n))$. This proves the upper bound.

  For the lower bound, suppose first that $n$ is of the form $4k+6$ and let $P \subseteq \bbR^3$ be a set $n/2=2k+3$ points in general position achieving the maximum number $h_3(2k+3)$ of halving planes. Let $P'$ be a copy of $P$ translated by a sufficiently small distance $\epsilon > 0$ in a generic direction. We then slightly perturb the points of $P \cup P'$ to ensure general position. For a point $p \in P$, we denote by $p'$ its copy (or \emph{twin}) in $P'$. Consider the sets $R^*$ (red) and $B^*$ (blue) of $k$ points each obtained as follows: for each pair $(p, p')$,
  we assign one point to $R^*$ and the other to $B^*$ uniformly and independently at random.

  Recall that a halving plane is defined by a triple of points and divides the remaining points as evenly as possible. Let $\pi_1$ be a halving plane of~$P$ defined by a triple $(a,b,c)$,
  so that $\pi_1$ has exactly $k$ points of~$P$ on each side.
  Consider the set $S \coloneqq 
  \{a,b,c,a',b',c'\}$. We claim that, with constant probability, there exists a plane $\pi_2$ that passes through a red point and a blue point of $S$ that are not twins, and has precisely one red and one blue point of $S$ on each side.
  
  As our perturbation $\epsilon$ is arbitrarily small, $\pi_2$ partitions the remaining points of $R^* \cup B^*$ in the same manner as $\pi_1$ partitions $P$. Let $R$ and $B$ be the sets of planes dual to points in~$R^*$ and~$B^*$, respectively.
  In the dual, the plane $\pi_2$ corresponds to a point $\pi_2^*$ on an edge of the arrangement $\AAA(R \cup B)$
  that lies on level~$k$ of~$\AAA(R)$ and of~$\AAA(B)$. 
  Therefore it lies on the curve~$L$ of intersection of two $k$-levels and this curve contains all the points $\pi_2^*$ (as we range over halving planes $\pi_1$ of~$P$). 
  
  Since different halving planes correspond to partitioning $P$ in different ways, the planes $\pi_2$ partition $P \cup P'$ in different ways and, hence, points $\pi_2^*$ arising from different halving planes of $P$ lie on different edges of  $\AAA(R \cup B)$.
  By construction, the number of edges of $L$ is $\Omega(h_3(n))$ in expectation, completing the lower bound proof for $n = 4k + 6$.

  Finally, for a general $n$, we write $n=4k+6+c$, for $0\leq c \leq 5$, we apply the above construction to $n - c$ and, at the end, add $c$ planes in general position to the set $R$ lying above all vertices of $\AAA(R \cup B)$.  It is easily checked that this addition does not reduce the size of~$L$. 
\end{proof}

Note, however, that in our application, the sets of points $R^*$ and $B^*$ are strictly separated, which is not the case in the above lower-bound construction.  We show that, with this additional constraint, the complexity of $L$ is $O(nh_2(n))$ and that there exist pairs of separated point sets $R^*$ and $B^*$ that achieve this bound.

\begin{fact}
\label{fact:intersectioncomplexity2}
  Let $\ell^*(n)$ be the maximum complexity of the intersection of two levels in $\AAA(R)$ and $\AAA(B)$ as in Fact~\ref{fact:intersectioncomplexity}, with the additional constraint that  the dual sets $R^*$ and $B^*$ are strictly separated by a plane.  Then $\ell^*(n)=\Theta(n h_2(n))$.
\end{fact}

\begin{proof}

\emph{Upper bound}:  
Let $L$ be the intersection of level~$k$ in $\AAA(R)$ and level~$k'$ in $\AAA(B)$.
Without loss of generality, we consider vertices of $L$ that are the intersection of one plane in $R$ and two planes in $B$. Such a vertex has $k - 1$ planes of $R$ and either $k-1$ or $k-2$ planes of $B$ below it. 

We work with the dual point sets $R^*$ (red) and $B^*$ (blue), strictly separated by the plane~$\pi$. Here, a vertex as above corresponds to a plane passing through one red point and two blue points, with $k - 1$ red points and $k-1$ or $k-2$ blue points above it. We count all such planes $\sigma_a$ passing through a fixed red point $a \in R^*$. 

Let $B'$ be a set of points obtained by radially projecting points of $B^*$ onto $\pi$ with center~$a$. Let $\ell_a$ be the line $\pi \cap \sigma_a$. Since $R^*$ and $B^*$ are separated by $\pi$, $\ell_a$ has $k-1$ or $k-2$ points of $B'$ on one side of it. Hence, the dual of $\ell_a$ is a vertex on level $k$ (or its complement) in the planar arrangement of the dual of $B'$ in $\pi$. Therefore, there are at most $2g_2(|B'|)$ choices for the plane $\sigma_a$, and at most $|R|\cdot2g_2(|B|) = O(n h_2(n))$ such planes overall.

Planes passing through exactly one blue point and two red points are handled symmetrically, completing the proof of the upper bound.

  \smallskip

  \noindent
  \emph{Lower bound}: Once again, it is sufficient to make the argument for $n$ of the form $4k+3$ for a positive integer $k$; the general case is handled as in the lower bound proof of Fact~\ref{fact:intersectioncomplexity}.  Consider a set $R^*$ of $2k + 2$ points realizing $h_2(2k+2)$ lying in the $xy$-plane, scaled to fit in the rectangle $(0,1)\times(0,\epsilon)$, for a sufficiently small $\epsilon>0$.

  Let $B^*$ be a set of $2k+1$  points equally spaced on a unit circle in the plane $x=0$, centered at the origin, so that no point lies in the $xy$-plane. Note that $R^*$ and $B^*$ are separated by the plane $x=\delta$, for a sufficiently small $\delta>0$.

  By making $\epsilon$ small enough, we can ensure that any halving line of $R^*$ passes arbitrarily close to the origin.  In particular, any pair of a halving line of $R^*$ and a point of $B^*$ define a halving plane that passes through two points of $R^*$ and one of $B^*$, and has exactly $2k$ points of $R^* \cup B^*$ on each side.  The number of such halving planes is $(2k+1)h_2(2k+2)=\Omega(n h_2(n))$.
  Finally, the points of $R^* \cup B^*$ can be perturbed to satisfy the general position assumption without reducing this number.

  So we have constructed two separated sets of points $R^*$ and $B^*$ with the property that there are $\Omega(n h_2(n))$ halving planes spanned by three of the points that simultaneously bisect both sets. 
  The lower bound follows by considering the dual sets $R$ and $B$.
\end{proof}

\section{The algorithm}

We can deduce the existence of eight-partitions of a finite point set $P \subset \bbR^3$ of a certain advantageous form from Theorem~\ref{thm:hadwiger}. 
  
\begin{observation}
  \label{obs:mass-8-partition}
  Let $k > 0$ be an integer and $P  \subseteq \bbR^3$ be a set of $n=8k+7$ points in general position.
  Then, there exists a triple of planes $(H_1,H_2,H_3)$ that eight-partitions $P$ with the following properties:
  \begin{enumerate}[(i)]
  \item \label{cont-h1} $H_1$ is horizontal (i.e., parallel to the $xy$-plane) and passes through the $z$-median point of~$P$. From here on, we refer to the $4k+3$ points that lie below (resp., above) $H_1$ as \emph{red} (resp., \emph{blue}) points and denote the sets $R$ (resp. $B$).
  \item \label{cont-k} $H_2$ and $H_3$ each contain exactly three points, and each open octant contains exactly $k$ points.
  \item \label{cont-h2h3} $H_2, H_3$ each bisect $R$ and $B$, and the pair $(H_2, H_3)$ four-partitions both $R$ and $B$. Furthermore, $H_2$ and $H_3$ contain at least one point of each color.
  \end{enumerate}
\end{observation}

\begin{proof}
  Since the set $X \coloneqq \{ (H_1, H_2, H_3) : H_1 \text{ is horizontal} \} \subset (\bbS^3)^3$ is compact, by Theorem~\ref{thm:hadwiger} and Lemma~\ref{lem:limiting_lemma}, there exists a configuration $\ccH_\infty = (H_1, H_2, H_3)$ that eight-partitions the point set with $H_1$ horizontal. Along with the general position assumption, this implies that $H_1$ contains only the $z$-median point. This proves (\ref{cont-h1}).

  To see (\ref{cont-k}), note that any eight-partition has at most $k$ points of $P$ in each of the eight open octants, one point in $H_1$, and at most three points in each of $H_2$ and $H_3$, by general position, for a total of at most $8k+1+2\cdot3=8k+7=n$ points.  So, in fact, all the inequalities are equalities: there \emph{must} be \emph{exactly} $k$ points in each open quadrant and exactly three points of $R \cup B$ in each of $H_2$ and $H_3$.

  It remains to show (\ref{cont-h2h3}). By the preceding paragraph, it is straightforward to see that $(H_2, H_3)$ four-partitions both $R$ and $B$. By Corollary~\ref{cor:all_partitions}, we have that any pair $(H_i, H_j)$ four-partitions $P$. Since $(H_1, H_2)$ four-partitions $P$, each quadrant formed by $(H_1,H_2)$ has at most $\ceil{(8k+7)/4} = 2k + 1$ points. 
  $H_1$ has $4k+3$ points on each side. Hence, we obtain that $H_2$ bisects $R$ and $B$, and, in particular, contains at least one point of each color. A symmetric argument shows that $H_3$ bisects both $R$ and $B$, and contains at least one point of each color.  This completes the proof.
\end{proof}

\begin{theorem}[Computation of an eight-partition]
  \label{thm:alg-detailed}
  Let  $P \subseteq \bbR^3$ be a set of $n > 0$ points in general position and $v \in \bbS^2$. An eight-partition $(H_1, H_2, H_3)$ of $P$, with $v$ being the normal vector of $H_1$,
  can be computed in time $O^*(n+\ell^*(n))$.
\end{theorem}
\begin{remark*}
   Since $\ell^*(n) = \Theta(nh_2(n)) = O(n^{7/3})$ by Fact~\ref{fact:intersectioncomplexity2}, we can compute an eight-partition in time $O^*(n^{7/3})$.
\end{remark*}
The rest of this section is devoted to the proof of Theorem~\ref{thm:alg-detailed}. We assume, without loss of generality, that $v = e_3 = (0,0,1)$ is the vertical vector, so $H_1$ is required to be horizontal.
If $n \le 7$, the statement holds trivially --- set $H_1$ to be the horizontal plane containing any point of $P$, and $H_2, H_3$ to contain at most three distinct points each, so that the octants do not contain any points. From here on, we will assume that $n = 8k+7$, for an integer $k > 0$. If $n$ is not of this form, we may add dummy points to $P$ (in general position) until the number of points is of the required form and run the algorithm. Once the algorithm terminates, we discard the dummy points, resulting in an eight-partition with at most $k$ points in each octant.

We now describe the algorithm to construct an eight-partition of $P$ satisfying the properties in Observation~\ref{obs:mass-8-partition}. Let $H_1$ be the horizontal plane containing the $z$-median point of $P$, and, without loss of generality, identify $H_1$ with the $xy$-plane. Now consider the sets $R$ and $B$ of $4k+3$ points each lying below and above, respectively, $H_1$. We further assume, without loss of generality, that $B$ is contained in the half-space $x < 0$ and $R$ is contained in the half-space $x > 0$. Otherwise, since no point in $R \cup B$ has $z = 0$ by the general position assumption, there exists a plane $H$ containing the $y$-axis and with sufficiently small negative slope in the $x$ direction such that all red points are below $H$ and all blue points are above $H$. Applying a generic sheer transformation (so as not to violate the general position assumption) that fixes the $xy$-plane and maps $H$ to the plane $x = 0$, we obtain point sets with the required properties.

Let $R^* = \{p^*: p \in R\}$ be the set of red planes dual to points in $R$ and set $\AAA(R) := \AAA(R^*)$ to be the arrangement formed by the set $R^*$. The set of blue planes $B^*$ and the blue arrangement $\AAA(B)$ are defined analogously. We will write $\AAA \coloneqq \AAA(R\cup B)$ for the arrangement formed by the planes in~$R^*\cup B^*$.  For a (dual) point $p \in \bbR^3$, we set $R^+_p, R^-_p \subseteq R^*$ to be the set of red planes lying strictly above and below $p$, respectively. For a pair $p, q$ of (dual) points, put
\[
  \funR(p,q) \coloneqq |R^+_p\cap R^+_q|-|R^+_p\cap R^-_q|-|R^-_p\cap R^+_q|+|R^-_p\cap R^-_q|.
\]
The sets $B^+_p, B^-_p \subseteq B^*$ and the function $\funB(p,q)$  are defined analogously for~$B^*$.

Let $L$ be the intersection of the median levels of $\AAA(B)$ and $\AAA(R)$. Let $m$ be the complexity, i.e., the number of vertices and edges, of $L$. By the following lemma, we have that $L$ is a connected $y$-monotone polygonal curve and is an alternating sequence of edges and vertices of $\AAA$ terminated by half-lines. 
\begin{lemma}
  \label{lem:L-monotone}
  $L$ is a connected $y$-monotone curve.
\end{lemma}
\begin{proof}
  Recall that $B$ lies in the quadrant $x<0, z<0$ and $R$ lies in the quadrant $x>0, z>0$.  Hence, the dual planes in $B^*$ and $R^*$ have equations of the form $z = ax + by + c$ with $a < 0$ and $a > 0$, respectively.

  Consider the intersection of $R^* \cup B^*$ with the plane $\Pi_d\colon y=d$. The intersection of the plane $z = ax + by + c$ is the line $z = ax + (bd + c)$, so planes in $B^*$ correspond to lines with negative slope and planes in $R^*$ correspond to lines with positive slope.
  In particular, the median levels of lines corresponding to $B^*$ and $R^*$ are graphs of piecewise-linear total functions  that are decreasing and increasing, respectively. It follows that the two curves intersect exactly once. This intersection point corresponds to the intersection of $L$ with the plane $\Pi_d$.

  By general position, $L$ is a union of vertex-disjoint cycles and bi-infinite paths composed of edges and vertices of $\AAA$, since incident to every vertex of $A$ contained in $L$ are precisely two edges belonging to $L$.  By monotonicity in $y$, $L$ must be a single bi-infinite chain.
\end{proof}

In fact, $L$ can be computed efficiently using standard tools~\cite{dynamic-AgMat,TMChan}, which we outline now.
\begin{lemma}[Computing the intersection of two levels \cite{dynamic-AgMat,TMChan}]\label{lemma:compute-L}
The intersection curve $L$ of the median level of $\AAA(B)$ and the median level of $\AAA(R)$ can computed in time $O^*(n+m)$, where $m$ is the complexity of the curve and $O^*(\cdot)$ notation hides polylogarithmic factors.
\end{lemma}
\begin{proof}
    We use the standard dynamic data structure for ray-shooting queries in the intersection of half-spaces; the currently fastest algorithm is due to Chan~\cite{TMChan}, see also earlier work of Agarwal and Matou{\v s}ek \cite{dynamic-AgMat}.

    A starting ray of $L$ can be computed by computing the intersection of the median levels in the vertical plane $\Pi_d\colon y = d$ for a small enough $d$, defined as in the proof of Lemma~\ref{lem:L-monotone}, in linear time, using an algorithm of Megiddo \cite{Meg}.
    
    Consider a point $p$ on the initial edge of $L$ (infinite in the $-y$-direction).  It lies on one plane of $\pi \in B$ and one plane $\pi' \in R$.
    Let $\ell$ be the intersection line of $\pi$ and $\pi'$, and consider the half line $\rho$ of $\ell$ starting at $p$ and infinite in the $+y$-direction. The planes of $B \cup R$ (besides $\pi$ and $\pi'$) are classified into those lying above $p$ and those lying below it.  Call the first set $U$ and the second $D$. We preprocess the intersection of the set of lower half-spaces defined by the planes of~$U$ and the intersection of the set of upper half-spaces defined by the planes of~$D$ for dynamic ray shooting and shoot with~$\rho$.  The earlier of the two intersections identifies the first plane $\pi_2$ of $(B\cup R) \setminus \{\pi,\pi'\}$ that $\rho$ meets.  This is the next vertex of $\AAA(B\cup R)$ on~$L$; $L$~turns here.  If $\pi_2$ belongs to $B$, $L$ now follows the intersection line of $\pi_2$ and $\pi'$.  Otherwise it follows the intersection line of $\pi$ and $\pi_2$.  Past the intersection, the sets $U$ and $D$ need to be updated and we continue, following the next ray, until we trace all of $L$.

    The only cost besides the initial computation of $p$ are identifying $U$ and $D$ in $O(n)$ time, initializing the dynamic structure, in $O^*(n)$ time, and performing two ray shots and $O(1)$ updates on $U$ and $D$ per vertex of $L$, each at a cost of $O^*(1)$.
\end{proof}

\begin{note*} 
  As we construct $L$, we can store
  it as a sequence of vertices and edges.  Each edge is associated with the red-blue pair of planes containing it.  An endpoint of an edge is contained in an additional plane.  For each consecutive edge/vertex pair $(e,v)$, in either direction, we record which new plane contains $v$ together with its color.
\end{note*}

We now return to the computation of the eight-partition $(H_1,H_2,H_3)$.  By the general position assumption, $H_2$ and $H_3$ cannot be vertical, so $H_2$ and $H_3$  correspond to vertices in~$\AAA$, by Observation~\ref{obs:mass-8-partition}. With the above setup, we can reformulate the problem of computing $H_2$ and $H_3$ as follows.
\begin{claim}[The dual alternating sign functions]
  \label{cl:XY}
  Computing $H_2$ and $H_3$ is equivalent to identifying a pair of vertices $p, q \in L$ such that $\funB(p,q)=\funR(p,q)=0$. 
\end{claim}

\begin{proof}

By Observation~\ref{obs:mass-8-partition}(\ref{cont-k}), the eight-partition $(H_1, H_2, H_3)$ has exactly $k$ points in each of the eight open octants. Setting $p \coloneqq H_2^*$ and $q \coloneqq H_3^*$, we obtain that $|R^\pm_p\cap R^\pm_q|=|B^\pm_p\cap B^\pm_q|=k$ for all combinations of signs. Therefore $\funB(p,q)=\funR(p,q)=0$, as claimed.  

We now argue the other direction. Let $p, q \in L$ be vertices such that $\funR(p, q) = \funB(p, q) = 0$. Since $p$ and $q$ lie on $L$, $H_2 \coloneqq p^*$ and $H_3 \coloneqq q^*$ bisect both $R$ and $B$ and contain exactly three points each, at least one of each color. Hence, it suffices to show that $(H_2, H_3)$ is a four-partition of both $R$ and $B$, i.e., $|R_p^\pm \cap R_q^\pm|, |B_p^\pm \cap B_q^\pm| \leq k$ for all combinations of signs. Indeed, this implies that each octant formed by $(H_1, H_2, H_3)$ contains exactly $k$ points, completing the proof.

Let $a_r \coloneqq |R^+_p\cap R^+_q|$, $b_r \coloneqq |R^+_p\cap R^-_q|$, $c_r\coloneqq |R^-_p\cap R^+_q|$, and $d_r\coloneqq |R^-_p\cap R^-_q|$. Define $a_b, b_b, c_b, d_b$ analogously for the blue planes. Without loss of generality, for a contradiction, suppose $a_r>k$.

  We first consider the case $a_r \geq k+2$. Since $p$ lies on the median level of $\AAA(R)$, we have $a_r + b_r \leq |R_p^+| = 2k+1$, implying $b_r \leq k-1$. Similarly, since $q$ lies on the median level of $\AAA(R)$, we have $c_r \leq k-1$.
  Recall that, by assumption, $\funR(p, q) = a_r+d_r-b_r-c_r= 0 $, implying $d_r = b_r + c_r - a_r \leq k-4$.  
  Hence, $a_r + b_r + c_r + d_r \leq 4k-4$, so $p$ and $q$ together are contained in $4k+3 - (a_r + b_r + c_r + d_r) \geq 7$ red planes, contradicting the general position assumption. 

  We may now assume $a_r = k+1$. Following the same reasoning we obtain $b_r \leq k$, $c_r \leq k$, and $d_r=b_r+c_r-a_r\leq k-1$. This implies $a_r + b_r + c_r + d_r \leq 4k$, and, in particular, that $p$ and $q$ together are contained in at least 3 red planes.
  Now consider the blue planes and note that $a_b + b_b + c_b + d_b \leq 4k$ --- this is clear if each of sets $B^\pm_p \cap B^\pm_q$ contains at most $k$ blue planes, otherwise it follows by the same argument as above. Hence, $p$ and $q$ together are contained in $4k+3 - (a_b + b_b + c_b + d_b) \ge 3$ blue planes.

  By Observation~\ref{obs:mass-8-partition}(\ref{cont-k}), $p$ and $q$ are contained in at most $6$ planes of $R^* \cup B^*$. Combined with the argument above, this implies $p$ and $q$ together are contained in exactly 3 planes of each color. It follows that $a_r + b_r + c_r + d_r = a_b + b_b + c_b + d_b = 4k$, which, by the assumption $a_r = k+1$, implies $b_r = c_r = k$ and $d_r = k-1$. Since $|R^-_p| = 2k+1$ and $b_r + d_r = 2k-1$, there are exactly 2 red planes containing $q$ below $p$. Similarly, since $|R^-_q| = 2k+1$ and $b_r + d_r = 2k-1$, there are exactly 2 red planes containing $p$ below $q$. But then $p$ and $q$ are contained in a total of $4$ red planes, a contradiction.

  This exhausts all possibilities and, hence,  $|R_p^\pm \cap R_q^\pm|, |B_p^\pm \cap B_q^\pm| \leq k$ for all combinations of signs, completing the proof.
\end{proof}

To summarize, once we construct $L$ in time $O^*(n+m)$, to compute an eight-partition, it is sufficient, by Claim~\ref{cl:XY}, to find two vertices $p,q \in L$ satisfying $\funR(p,q) = \funB(p,q) = 0$. 
In particular, it is possible to construct an eight-partition by enumerating all the $\Theta(m^2)$ pairs of vertices in $L$; the exact running time depends on how efficiently one can check candidate pairs.
Below, we describe
how to reduce the amount of remaining work to $O((m+n)\log m)$.

\subparagraph*{Speed up}
For simplicity of later calculations, we orient $L$ in the $y$-direction (which is possible by Lemma~\ref{lem:L-monotone}) and view it as an alternating sequence of edges and vertices, starting and ending with a half-line. We denote these elements by $x_1, x_2, \dotsc, x_m$. Recall that the goal is to identify $i,j \in [m]$ so that $x_i, x_j$ are vertices and $\funR(x_i,x_j)=\funB(x_i,x_j) = 0$. 

We extend the definition of $\funR, \funB$ as follows.
If $x_i, x_j$ are both edges, we pick arbitrary points $p$ and $q$ in the open edges $x_i$ and $x_j$, respectively, and
set $\funR(x_i,x_j) \coloneq \funR(p, q)$ and $\funB(x_i,x_j) \coloneq \funB(p, q)$; the cases where $x_i$ is an edge or $x_j$ is an edge, but not both, are handled analogously.
Note that specifying the (open) edges containing $p$ and $q$ is sufficient to determine $\funR$ and $\funB$, hence the definition is unambiguous. Define $\pi\colon [m]^2 \to \bbZ^2$ by
\[
  \pi(i,j)\coloneq(\funR(x_i,x_j),\funB(x_i,x_j)).
\]
With this setup, our goal is to identify a point $(i, j) \in [m]^2$ (corresponding to a pair of vertices on $L$) such that $\pi(i, j) = \bO$.

We define a \emph{grid curve} $C$ to be a sequence of points $(i_1, j_1),\dots,(i_t,j_t)$ in $\bbZ^2$ such that $(i_{\ell+1},j_{\ell+1}) \in \{ (i_\ell,j_\ell), (i_{\ell\pm 1},j_\ell),  (i_\ell,j_{\ell\pm 1})\}$ for each $\ell \in [t-1]$.
In words, a grid curve is a walk in $\bbZ^2$ which, at each step, does not move at all or moves by exactly one unit up/down/left/right.
A curve is \emph{closed} if $(i_1,j_1)=(i_t,j_t)$. A grid curve is \emph{simple} if non-consecutive points are distinct (we think of the start and end points as consecutive) --- so the curve does not revisit a point after it moves away from the point.

To each grid curve $C$, we associate a piecewise linear curve $\lineC$ in $\bbR^2$, consisting of line segments connecting consecutive points $(i_{\ell},j_{\ell}), (i_{\ell+1},j_{\ell+1})$ of $C$ for each $\ell \in [t - 1]$. For a curve $\lineC$ not passing through the origin $\bO$, the winding number $w(\lineC)$ about $\bO$ is defined in the standard way. Slightly abusing notation, we set $w(C) := w(\lineC)$.
In particular, provided $\lineC$ misses the origin,
\[ w(C) = w(\lineC) = \begin{cases*}
0  & if $\lineC$ does not wind around $\bO$, \\
n > 0 & if $\lineC$ winds around $\bO$ $n$ times counterclockwise,\\
n < 0 & if $\lineC$ winds around $\bO$ $-n$ times clockwise.
\end{cases*} \]
We omit the rigorous definition of $w(C)$ as a contour integral in the complex plane since it does not add to the discussion and, instead, refer the reader to \cite[Chapter 4.4.4]{krantz1999handbook}.

Our algorithm proceeds as follows: 
\begin{enumerate}[Step 1]
  \item Set $C \coloneq T$ (see Definition~\ref{def:T}). If $\pi(C)$ meets $\bO$, then stop --- we have found a point that maps to $\bO$ (see Lemma~\ref{fact:hit-it}).  Otherwise $\pi(\lineC)$ has odd winding number, by Lemma~\ref{fact:T}.
  \item\label{algstep2} Construct two simple closed curves $C_1$, $C_2$ so that (a)~$\lineC = \lineC_1 + \lineC_2$, (b)~at least one of $\pi(C_1), \pi(C_2)$ has odd winding number (unless they meet $\bO$), (c)~the regions enclosed by $\lineC_1$ and $\lineC_2$ partition the region enclosed by $\lineC$, and (d)~the area enclosed by each of $\lineC_1, \lineC_2$ is a fraction of that enclosed by $\lineC$ (see Lemma~\ref{fact:curve-splitting}).
  \item If $\pi(C_1)$ or $\pi(C_2)$ meets $\bO$, then stop --- we found a point that maps to $\bO$, by Lemma~\ref{fact:hit-it}. 
  \item Compute $w(\pi(C_1))$ and $w(\pi(C_2))$, and replace $C$ with the one with the odd winding number. Goto Step \ref{algstep2}.
\end{enumerate}

We now proceed to fill in the details, starting with the definition of the initial curve $T$.
\begin{definition}[The triangular grid curve $T$]
\label{def:T}
The simple closed grid curve $T$ traverses a \emph{triangular} path defined as follows:
\begin{itemize}
 \item Starting with the bottom horizontal side of the grid $[m]^2$, $T$ traverses the points 
 \[
   (x_1, x_1), (x_2, x_1), \dots, (x_m, x_1),
 \]
 \item continuing along the right vertical side of the grid $[m]^2$ along the points
 \[
   (x_m, x_1), (x_m, x_2), \dots, (x_m, x_m),
 \]
 \item finally, traversing back \emph{diagonally} along
 \[
   (x_m, x_m), (x_{m-1}, x_{m}), (x_{m-1}, x_{m-1}), (x_{m-2}, x_{m-1}), \dots, (x_1, x_2), (x_1, x_1).
 \]
\end{itemize}
\end{definition}
Along the diagonal side of $T$, we are really only interested in points of the form $(x_\ell, x_\ell)$ with $\ell \in [m]$. However, since this doesn't give a grid curve, we ``patch'' it up by introducing intermediate points. Fortunately, this does not change the desired properties of $T$.

\begin{lemma}
  \label{fact:grid-to-grid}
  If $C$ is a grid curve, then $\pi(C)$ is a grid curve.  Moreover, if $L$ has already been computed, $\pi(C)$ can be computed in time $O(n+|C|)$.
\end{lemma}
\begin{proof}
  Consider a step in $C$ from $(x_i,x_j)$ to $(x_{i+1},x_j)$, where $x_i$ is an edge of $L$ and $x_{i+1}$ is a vertex. Then $x_{i+1}$ is contained in the planes that contain $x_i$ and one additional plane $H$. Suppose, without  loss of generality, that $H$ is red. This means that the cardinality of one of the sets $R_p^\pm$ changes by one. Hence, the cardinality of $R_p^\pm \cap R_q^\pm$, for each combination of signs, changes by at most one --- if $H$ contains $q$, nothing changes. It follows that the function $\funR$ changes by at most one, and the function $\funB$ remains unchanged.

  Note that, up to symmetry, only one such transition or its reverse occurs in a single step of $C$. We've shown that each step causes either $\funR$ or $\funB$ (but not both) to change by at most one, and, hence, $\pi(C)$ is a grid curve.

  The computation can be carried out in constant time per incident edge-vertex pair of $C$, since $L$ has been already computed, after a $O(n)$-time initialization that computes $\funR,\funB$ at an arbitrary starting point of $C$ by brute force.
\end{proof}

Lemma~\ref{fact:grid-to-grid} immediately implies the following.
\begin{lemma}
  \label{fact:hit-it}
  If $\overline{\pi(C)}$ meets $\bO$, then some point of $C$ is mapped to $\bO$.
\end{lemma}

A key property of the triangular grid curve $T$ is the following.
\begin{lemma}
  \label{fact:T}
If $\bO \not\in \pi(T)$, then $w(T)$ is odd.
\end{lemma}
\begin{proof}
  Let $N \coloneqq 4k+2$, and let $H, V, D$ be the images (under $\pi$) of the horizontal, vertical, diagonal sides of $T$, respectively. Note that $\pi(T)$ is the concatenation of $H, V,$ and $D$ in that order.
  
  Observe that if $p = q = x_i$ with $i \in [m]$, then $|R_p^+ \cap R_q^-| = |R_p^- \cap R_q^+| = 0$. Hence, $X(x_i, x_i) \in \{ 4k+1, 4k+2 \}$ depending on whether $x_i$ is contained in one or two red planes. Similarly, $Y(x_i, x_i) \in \{ 4k+1, 4k+2 \}$. Hence $\pi(x_i, x_i) \in \{ (N, N), (N-1, N-1)\}$ and, in particular, $\pi(x_i, x_i) = (N, N)$ if $x_i$ is an edge. Along with Lemma~\ref{fact:grid-to-grid}, this implies that the grid curve $D$ is a closed walk on the points in $\{N-2, N-1, N, N+1\}^2 \setminus \{ \bO \}$ starting and ending at the point $(N, N)$.
    
  Noting that $x_1$ and $x_m$ are half-lines contained in the same red plane, and that every red plane that lies above $x_1$ lies below $x_m$ and vice versa, we obtain $\pi(x_m, x_1) = (-N, -N)$. Hence, $H$ is a grid curve from the point $(N, N)$ to $(-N, -N)$ and $V$ is a grid curve from the point $(-N, -N)$ to $(N, N)$.
   
  The discussion above implies that $w(T)$ is equal to the winding number of 
  the concatenation of $V$ and $H$. We argue below that $V$ is the image of $H$ under a rotation by 180\textdegree{} around the origin, i.e., the map $(x,y) \mapsto (-x,-y)$. 
  Since, by assumption, neither $H$ nor $V$ contain $\bO$, the concatenation of $H$ and $V$ has odd winding number as claimed.

  Specifically, we need to show that $\pi(x_i, x_1) = -\pi(x_m, x_i)$. Since $\pi$ is symmetric in the two arguments, it suffices to show that $\pi(x_1, x_i) = -\pi(x_m, x_i)$. As mentioned before, every plane that lies above $x_1$ lies below $x_m$ and vice versa. That is, $R_{x_1}^+ = R_{x_m}^-$ and $R_{x_1}^- = R_{x_m}^+$, and similarly $B_{x_1}^+ = B_{x_m}^-$ and $B_{x_1}^- = B_{x_m}^+$. The claim is now obvious from the definition of the functions $\funR$ and $\funB$. 
\end{proof}

\begin{lemma}
\label{fact:promisingmapstozero}
  If $w(\pi(C))$ is odd, then there is a point $(i, j)\in\bbZ^2$ enclosed by $\lineC$ with $\pi(i, j) = \bO$.
\end{lemma}

\begin{proof}
A \emph{grid square} $S$ is a simple closed grid curve of the form 
\[
  (i, j), (i+1, j), (i+1, j+1), (i, {j+1}), ({i}, {j})
\]
with $(i, j) \in \bbZ^2$.
A \emph{square} is $\overline{S}$ for some grid square $S$. If there is a grid square $S$ enclosed by $\lineC$ such that $\pi(S)$ meets $\bO$, then we are done by Lemma~\ref{fact:hit-it}. Otherwise, note that $\overline{\pi(C)}$ is the sum of the images of the corresponding squares. Hence, there is a grid square $S$ with $w(\pi(S))$ odd. By Lemma~\ref{fact:grid-to-grid}, $\pi(S)$ is a grid curve. By enumerating all possibilities (see Fig.~\ref{fig:square}), we conclude that $w(\pi(S))$ cannot be odd.
\end{proof}

\begin{figure}
  \centering
  \includegraphics[width=0.9\textwidth]{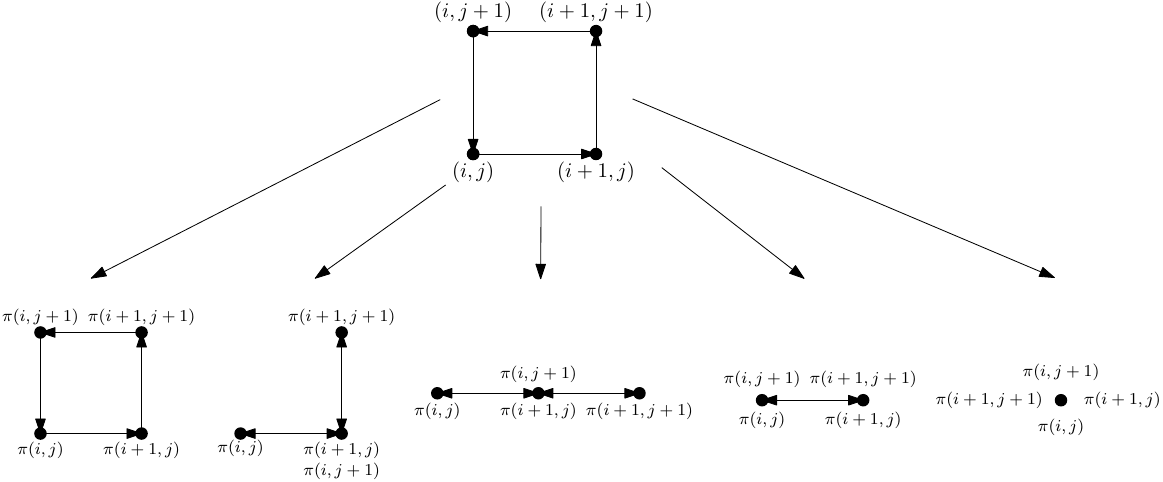}
  \caption{Up to symmetries, the different possibilities for the image under $\pi$ of a grid square $S$, which is always always a grid curve in $\bbZ^2$, by Lemma~\ref{fact:grid-to-grid}.  Note that the image cannot have odd winding number.}
  \label{fig:square}
\end{figure}

Next, we show how to decompose a curve $C$.  We restrict our attention to ``trapezoidal'' curves: Such a curve is the boundary of the intersection of the region bounded by the initial triangle $T$ with a grid-aligned rectangle.  This property is maintained inductively.
\begin{lemma}
\label{fact:curve-splitting}
  Given a trapezoidal curve $C$ whose image misses $\bO$, with $w(\pi(C))$ odd, we can construct two trapezoidal curves $C_1$ and $C_2$ so that
  \begin{enumerate}[(i)]
    \item \label{count:partitioning} the region $R$ surrounded by $\lineC$ is partitioned into region $R_1$ surrounded by $\lineC_1$ and region $R_2$ surrounded by $\lineC_2$.
    \item \label{count:reduction} $\area(R_1),\area(R_2)\le c \cdot \area(R)$, for an absolute constant $c<1$.
    \item \label{count:winding} either $\bO$ is in the image of $C_1$ and $C_2$ or $w(\pi(C))=w(\pi(C_1)) + w(\pi(C_2))$.
  \end{enumerate}
\end{lemma}

\begin{proof}
  We already noted that the image of a grid square cannot have odd winding number, therefore $R$ is not a grid square.
  As long as $R$ is at least two units high, divide it by a horizontal grid chord into two near-equal-height pieces (that is, the two parts have equal height, or the lower one is one smaller)
  producing two regions $R_1$ and $R_2$. The curves $C_1$ and $C_2$ are the boundaries of the regions (refer to Fig.~\ref{fig:regions}).  If the height of $R$ is one, perform a similar partition by a vertical chord into to near-equal-width pieces.
\begin{figure}
  \centering
  \includegraphics[width=.5\textwidth]{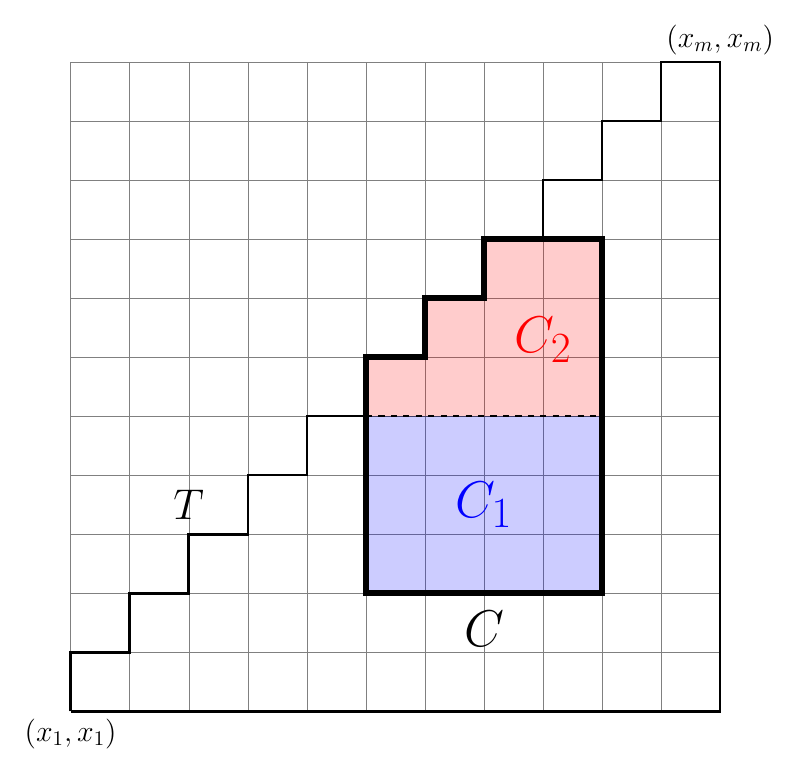}
  \caption{Curve $C$; the blue region is bounded by $C_1$, and the red by $C_2$, with the horizontal dividing chord drawn dashed.} 
  \label{fig:regions}
\end{figure}

  Property (\ref{count:partitioning}) is satisfied by construction.  If the image of the new chord misses $\bO$, then both $C_1$ and $C_2$ avoid $\bO$ and property (\ref{count:winding}) follows from the properties of the winding number on the plane.  Finally, an easy calculation shows that, if the split height/width is even, then each part contains at most $3/4$ of the original area; this fraction can rise to $5/6$ if $R$ has odd width or length (the extreme case is achieved at width/length of three), which proves property~(\ref{count:reduction}).
\end{proof}

\begin{lemma}
  \label{fact:winding}
  Given a simple closed grid curve $C$ in $[m]^2$ we can determine whether $\pi(C)$ contains a zero.  If not, we can compute $w(\pi(C))$, all in time $O(|C|+n)$.
\end{lemma}
\begin{proof}
  By Lemma~\ref{fact:grid-to-grid},
  we can trace $\pi(C)$ step by step and, in particular, detect whether $\bO\in\pi(C)$. So suppose this is not the case.

  Consider the (open) ray $\rho$ from the origin directed to the right in~$\bbZ^2$. To determine the winding number
  of the curve $\overline{\pi(C)}$ not passing through the origin, it's sufficient to count how many times the curve crosses the ray $\rho$.  
  We can compute the number of times $\overline{\pi(C)}$ crosses $\rho$ by computing $\pi$ for every vertex of $C$ in order and counting the number of times $(X,0)$ occurs along it, with $X>0$.
  
    As $\overline{\pi(C)}$ may partially overlap~$\rho$, we need to check whether $\overline{\pi(C)}$ arrives at $(X,0)$ with $X>0$ from below the $\funR$-axis and (possibly after staying on the axis for a while) departs into the region above $\funR$-axis, or vice versa.  That would count as a signed crossing.  Arriving from below and returning below, or arriving from above and returning above, does not count as a crossing.

  All of the above calculations can be done in time $O(1)$ per step of $\pi(C)$, after proper initialization, by Lemma~\ref{fact:grid-to-grid}.
\end{proof}

\subparagraph*{Running time}
We now analyze the running time of the algorithm we described. We can traverse a length-$O(m)$ closed grid curve $C$, compute its image $\pi(C)$, and check whether it passes through the origin in time $O(m+n)$ by Lemma~\ref{fact:grid-to-grid}.  One can check whether $\pi(C)$ winds around the origin an odd number of times by Lemma~\ref{fact:winding}, also in time $O(m+n)$.

The number of rounds of the main loop of the algorithm is $O(\log m)$, as the starting curve cannot enclose an area larger than $O(m^2)$ and areas shrink by a constant factor in every iteration, by Lemma~\ref{fact:curve-splitting}.
Combining everything together, we conclude that $L$ can be computed in $O^*(n+m)$ time, and the algorithm  can then identify the pair of vertices of $L$ corresponding to an eight-partition in at most $O(\log m)$ rounds, each costing at most $O(m+n)$.
This concludes the proof of Theorem~\ref{thm:alg-detailed}.

\bibliography{ref-pure.bib}

\appendix
\section{Limit arguments}\label{app:mass_facts}

We prove some standard facts using limit arguments.

\begin{lemma}[Limit argument for finite point sets]
  \label{lem:limiting_lemma}

  Let $X\subseteq \left(\bbS^3 \right)^3$ be a compact subset such that, for all $\mu$ mass distributions (with connected support) on $\bbR^3$ there is a plane configuration $\ccH\in X$ that eight-partitions $\mu$; then for any set $P$ of points in $\bbR^3$, there is a configuration $\ccH_\infty \in X$ that eight-partitions the point set. 
\end{lemma}
\begin{proof} Let $n$ be the number of points in $P$.
  Let $\mu_i$ be the measure defined as follows. At every point in $P$, place a ball of radius $\epsilon_i = \frac{1}{2^i}$ with uniform density and total mass $(1-\epsilon_i)/n$; finally, add a normal Gaussian distribution ``on the background,'' with total mass $\epsilon_i/n$.
  Note that the total measure $\mu_i$ of the complement of the union of balls is less than $\epsilon_i$.

  By choosing $i$ large enough, we can assume that a plane can intersect a collection of balls only if their centres are coplanar; hence, without loss of generality, we can assume that this happens for $i=1$.

  By assumption, there is a plane configuration $\ccH_i$ that eight-partitions the mass $\mu_i$ for each~$i$; by compactness of $X$, there is a subsequence $\ccH_{i_j}$ that converges to some limit $\ccH_\infty$; up to reindexing we can assume that the original sequence $\ccH_i$ does. The obtained limit point eight-partitions the original point set $P$: in fact, there is a $i_0$ big enough such that for every orthant $\alpha\in \ztwo^3$  and any $m\geq i_0$, every point $p\in P \cap \ccO^{\ccH_\infty}_\alpha$ is ``far away'' (e.g., at least $1/{2^{i_0}}$) from the planes in the configuration $\ccH_m$; hence
  \[
  \frac{1-\epsilon_m}{n}\lvert P\cap \ccO^{\ccH_\infty}_\alpha\rvert \leq \mu_m(\ccO^{\ccH_m}_\alpha) = \frac{1}{8}.  
\]
  Taking the limit in $m$ we obtain the desired result.
\end{proof}
\begin{corollary}
  \label{cor:all_partitions}
  Let $X$ and $P$ as above. If $\ccH_\infty = (H_1, H_2, H_3)$ is the configuration constructed in the proof of Lemma~\ref{lem:limiting_lemma}, then any plane $H_i$ in $\ccH_\infty$ bisects $P$ and any pair $(H_i, H_j)$ four-partitions the points.
\end{corollary}

\begin{proof} For simplicity we show the result for the first plane $H_1$ and the pair $(H_1, H_2)$, all the other cases are identical.
  First, construct $\mu_i$ and $\ccH_i = (H_{i, 1}, H_{i, 2}, H_{i, 3})$ converging to the limit $\ccH_\infty$ as in the proof of Lemma~\ref{lem:limiting_lemma}. 
  
  Again, by choosing $i_0$ big enough we obtain that, for any $m\geq i_0$ and any sign $\alpha\in \ztwo$ every point $p\in P\cap H_1^\alpha$ is sufficiently far form $H_{m, 1}$ hence
  \[
    \frac{1-\epsilon_m}{n}\lvert P\cap H_1^\alpha\rvert \leq \mu_m(H_{m, 1}^\alpha) = \frac{1}{2}.
    \]
    Similarly, for any pair of signs $(\alpha_1, \alpha_2)\in \ztwo^2$, any point $p\in P\cap H_1^{\alpha_1}\cap H_2^{\alpha_2}$ is sufficiently far from both $H_{m, 1}$ and $H_{m, 2}$, therefore
    \[
      \frac{1-\epsilon_m}{n}\lvert P\cap H_1^{\alpha_1}\cap H_2^{\alpha_2}\rvert \leq \mu_m(H_{m, 1}^{\alpha_1}\cap H_{m, 2}^{\alpha_2}) = \frac{1}{4}.
      \]
      By taking the limit we obtain the desired result.
    \end{proof}
    
    \begin{lemma}[Limit argument for mass distributions with possibly disconnected support]
      \label{lem:con_support}
      Let $X$ be a compact set in $\left(\bbS^3\right)^3$ such that, for any mass distribution with connected support there is a configuration $\ccH\in X$ that eight-partitions the measure. Let $\mu$ be a ``general'' mass distribution. Then there is a plane arrangement $\ccH_\infty\in X$ that eight-partitions $\mu$.
    \end{lemma}
    \begin{proof}
      Define $\epsilon_i \coloneq \frac{1}{2^i}$ and let $\mu_i$ the measure defined, on a measurable set $A\subseteq \bbR^3$, as
      \[
        \mu_i(A)\coloneqq \left(1-\epsilon_i\right)\mu(A) + \epsilon_i \mathcal{N}(A),
      \]
      where $\mathcal{N}$ is a normal Gaussian distribution on $\bbR^3$. Then, $\mu_i$ is a mass distribution with connected support hence there is a configuration $\ccH_i$ that eight-partitions $\mu_i$.  By compactness, up to taking a subsequence, $\ccH_i$ converges to a configuration $\ccH_\infty$. 
      
      Now, for any $\alpha\in \bbZ_2^3$, we have that
      \[
        \left(1-\epsilon_i \right)\mu(\ccO^{\ccH_i}_\alpha) \leq \mu_i(\ccO^{\ccH_i}_\alpha) = \frac18.
      \]
    
      For any fixed measure $\tilde \mu$, the map $\ccH \to \tilde \mu(\ccO^{\ccH}_\alpha)$ is continuous; hence by taking the limit in $i$ we obtain that, for all $\alpha\in \bbZ_2^3$,
      \[
        \mu(\ccO^{\ccH_\infty}_\alpha) \leq \frac18.
      \]
      However, since
      \[
        \sum_{\alpha\in \bbZ_2^3} \mu(\ccO^{\ccH_\infty}_\alpha) = \mu(\bbR^3) = 1,
      \]
      it follows that all the previous inequalities are equalities.
    \end{proof}

\section{Four-partitioning in the plane with a prescribed bisector}
\label{app:planepartitions}

This section is devoted to the proof of Lemma~\ref{lem:2dpartition}.  For convenience, we restate the lemma here. Both the lemma and proof are due to~\cite{blag_karasev_2016}. 
\blagkarasev*
\begin{proof}
Suppose, without loss of generality, that $v = (0, 1)$. We first prove existence. Let $\alpha \in [0, \pi/2]$, and  rotate $v$ counterclockwise and clockwise by angle $\alpha$ to obtain $u_\alpha$ and $w_\alpha$, respectively. Then, since the measure has connected support, there exist unique lines $\ell_\alpha$ and $m_\alpha$ that are perpendicular to $u_\alpha$ and $w_\alpha$, respectively, and bisect $\mu$. Note that $v$ bisects the angle between $\ell_\alpha$ and $m_\alpha$.

The (oriented) lines $\ell_\alpha$ and $m_\alpha$ partition the plane into four octants, which we label $P_N , P_E , P_S , P_W$ (north, east, south, west) in the obvious manner. Since both lines are bisecting, we have
\[
\mu(P_N) = \mu(P_S) = x,\quad \mu(P_W) = \mu(P_E) = \mu(\bbR^2)/2 - x = y.
 \]
When $\alpha$ tends to 0, $P_W$ and $P_E$ tend to empty sets and evidently $x > y$ for $\alpha$ sufficiently close to $0$.
When $\alpha$ tends to $\pi/2$, $P_N$ and $P_S$ tend to empty sets and then $x < y$ for $\alpha$ sufficiently close to $\pi/2$.
Since $x$ depends continuously on $\alpha$, we must have $x = y$ somewhere in between, by the intermediate value theorem. Thus, we have existence.

We now show uniqueness. Assume we have a partition $P_N, P_E, P_S, P_W$ with angle $\alpha$ and another partition $Q_N, Q_E, Q_S, Q_W$
with angle $\alpha'$. Assume without loss of generality that $\alpha' \leq \alpha$. Moreover, since for $\alpha' = \alpha$ the partition is defined uniquely, we may assume $\alpha' < \alpha$. Let $p = \ell_\alpha \cap m_\alpha$ and consider the following cases:
\begin{enumerate}
\item $p \in Q_N$: In this case $P_N \subset Q_N$ and both sets have the same measure. This contradicts connectivity of the set where the density is positive, since the density is positive in $Q_S$ and in $P_N$, it must be positive somewhere in $Q_N \setminus P_N$, implying $\mu(Q_N) > \mu (P_N)$.
\item $p \in Q_E$: In this case $P_W \subset Q_W$ and we obtain a similar contradiction.
\item $p \in Q_S$ and $p \in Q_W$ are similar to considered cases.
\end{enumerate}
Since the lines $\ell_\alpha, m_\alpha, \ell_{\alpha'}$, and $m_{\alpha'}$ are distinct, this covers all cases. In each case, we obtain a contradiction, hence, we have uniqueness.

As for continuity, it follows from the standard fact that a map with compact codomain and closed graph must in fact be continuous. The codomain is compact since we are only interested in directions of the halving lines that afterwards produce halving lines continuously, thus working with $\bbS^1 \times \bbS^1$ as the space of parameters.
\end{proof}

\end{document}